\documentclass[]{article}
\usepackage{times}
\usepackage{amssymb,amsfonts,latexsym,stmaryrd}
\usepackage{QED}
\usepackage[PostScript=dvips]{diagrams}

\usepackage[utf8]{inputenc}

\usepackage{enumerate}

\usepackage{boxedminipage}

\usepackage{ifpdf}

\ifpdf
\usepackage[pdftex]{graphicx}
\else
\usepackage{graphicx}
\fi
\title{An Algebraic Characterisation of Concurrent Composition}
\author{Samson Abramsky}

\begin{document}

\ifpdf
\DeclareGraphicsExtensions{.pdf, .jpg, .tif}
\else
\DeclareGraphicsExtensions{.eps, .jpg}
\fi

\newtheorem{theorem}{Theorem}[section]
\newtheorem{lemma}[theorem]{Lemma}
\newtheorem{definition}[theorem]{Definition}
\newtheorem{corollary}[theorem]{Corollary}
\newtheorem{remark}[theorem]{Remark}
\newtheorem{fact}[theorem]{Fact}

\newcommand{\Vops}{\textsf{Vops}}
\newcommand{\VFS}{\textsf{VFS}}
\newcommand{\PCL}{\ensuremath{\mathsf{PCL}}}
\newcommand{\CPCL}{\ensuremath{\mathbf{PCL}(\Sigma)}}
\newcommand{\PCVL}{\ensuremath{\mathsf{PCVL}}}
\newcommand{\CPCVL}{\ensuremath{\mathbf{PCVL}(\Sigma)}}
\newcommand{\LL}{\textsf{L}}
\newcommand{\And}{\ensuremath{\; \wedge \;}}
\newcommand{\Imp}{\; \Rightarrow \;}
\newcommand{\Def}{\mathsf{def}\;}
\newcommand{\IFF}{\;\; \Longleftrightarrow \;\;}
\newcommand{\Ac}{\mathsf{Ac}}
\newcommand{\dom}{\mathsf{dom} \;}
\newcommand{\PAlg}{\mathsf{PAlg}}
\newcommand{\CPAlg}{\mathbf{PAlg}}
\newcommand{\CSDPA}{\mathbf{SD{-}PAlg}}
\newcommand{\SDPA}{\mathsf{SD{-}PAlg}}
\newcommand{\UU}{\mathbf{U}}
\newcommand{\BB}{\mathbf{F}}
\newcommand{\GG}{\mathbf{G}}
\newcommand{\Bi}{\mathbf{F}'}
\newcommand{\empset}{\varnothing}
\newcommand{\pfn}{\rightharpoonup}
\newcommand{\id}{\mathsf{id}}
\newcommand{\Id}{\mathsf{Id}}
\newcommand{\proj}{\pi}
\newarrow{dash}....>
\newcommand{\myvec}{\underline}

\def\uL{\myvec{L}}
\def\us{\myvec{s}}
\def\ut{\myvec{t}}
\def\ua{\myvec{a}}
\def\ub{\myvec{b}}
\def\aL{\alpha{L}}

\maketitle

\section*{Prelude 2009}
I was invited by Olivier Danvy and Carolyn Talcott to contribute to this special issue of \textsl{Higher Order and Symbolic Computation} in honour of Peter Landin. They were keen to include some technical contributions, and I recalled my very first piece of theoretical work, which had only appeared as a 1981 Queen Mary College technical report, and which seemed as if it might be appropriate for several reasons:
\begin{itemize}
\item Firstly, it was written when I was a colleague of Peter's at Queen Mary.

\item Moreover, it was strongly influenced by his work, in particular by his paper `A Program-Machine Symmetric Automata Theory' \cite{L70}. This is essentially the last technical paper he published. It is very little known, and yet it seems to me to contain some striking and prescient ideas.
\end{itemize}
For these reasons, rather than any great merits of my fledgling effort, I tentatively suggested to Olivier and Carolyn that this old paper might make a suitable submission to the Special Issue. They were positive about the idea --- which left me with the task of retrieving this old document. Somewhat to my surprise, I found a copy of the report among my old papers --- probably the last surviving one. The present paper is that old report, essentially unchanged apart from a couple of comments, with this prelude added to provide some context.

What was the `Program Machine Symmetric Automata Theory' about?
The title itself is characteristically striking, as is the image with which the paper opens, of a ball sliding around a plane. Both the ball and the plane have graphs with coloured edges inscribed on them, and the motion is constrained to follow paths along both graphs, and also so that vertices and edge-colours are synchronized. One can think of the ball as a program, and the plane as the machine on which it runs. However, note that the situation is really quite symmetric; each constrains the other.

The paper goes on to generalize this situation to polygraphs with multi-edges, and to formalize it in the language of universal algebra. Both the ball and the plane --- the program and the machine --- are represented as algebras, and their interaction is captured as the reachable subalgebra generated in the product of the two algebras.

Let us recall that this paper appeared in 1969! While Petri had already published an early form of his Net theory, process calculi as we now know them did not appear until the late 1970's. When they did appear, as Robin Milner's CCS and Tony Hoare's CSP, it became possible to see that Landin's insight, and his formal construction, related directly to the issue of modelling the \emph{interaction of concurrent processes}. In fact, it can be seen as an algebraic characterization of process composition in a form which allows `true concurrency' to be captured. This is the main point made in my 1981 report, which follows. Thus Landin's insights were indeed well ahead of their time.

His paper is still well worth reading for its style, and the glimpses it affords of the cultivated imagination and deep insights of its author. There are also arresting asides, such as the footnote:
\begin{quotation}
For some years I have aspired to `language-free programming'. \ldots
\end{quotation}

I hope that some readers will be led back to this paper of Landin's. It deserves its place alongside the classics such as `A correspondence between Algol-60 and Church's $\lambda$-notation' and `The next 700 programming languages'.

\paragraph{Acknowledgements}
My thanks to Olivier Danvy and Carolyn Talcott for their support and encouragement.
Special thanks to  Tim McCarthy for his invaluable help in converting a scanned copy of an ancient technical report into Latex. 

\section{Introduction}


Much recent work in theoretical computer science has been directed to developing the semantics of concurrent systems and computations. We are concerned with two such efforts in particular: Shields' vector firing sequence semantics for path programs arising out of the work of Lauer, Shields and others at Newcastle on the COSY formalism \cite{SL79}; and Hoare's ``simple model for CSP'' \cite{H80}. THere is a close connection between the two approaches; indeed Hoare attributes his parallel composition operator to Campbell and Lauer (cf.~Acknowledgements in \cite{H80}), and it is equivalent to the earlier Firing Sequence semantics for COSY. Both approaches give a semantics for concurrent systems in terms of their overt behaviour, and express behaviour extensionally, in terms of sequences of event occurrences, the sequencing corresponding to order in time. The ``events'' in Hoare's model correspond to synchronised communication by processes, those in COSY to performances of operations; and the two approaches, one describing systems in terms of the agents from which they are synthesised, the other in terms of the constraints imposed on them, seem in some sense duals of each other. However, in this paper we consider the semantic constructions directly, and ignore the formalisms they may be used to interpret. The constructions are language-theoretic --- the ``sequences of event occurrences'' are represented simply by strings over an alphabet of event types --- and do not correspond to traditional language-theoretic constructions; in two extreme cases of its application, the Campbell-Lauer-Hoare combinator reduces to intersection, and the shuffle product, respectively, but in general it is a rather subtle combination of the two; while the vector firing sequence construction generates ``languages'' over a non-free monoid of \emph{vectors} of strings, with elementwise concatenation. Of the two constructions, vector firing sequences are undoubtedly a neater and more compact representation; they are a quotient of the other construction (in a sense to be made precise) and in effect provide a nice canonical representation of the quotient. However, they suffer the defect of not being homogeneous in their application, leading to bothersome questions of the dimensions of the vectors involved, etc.; and so do not lend themselves to a smooth algebraic treatment of the synthesis of complex systems out of simpler components.

Our treatment is aimed at providing an algebraic framework in which the advantages of compact representation, and homogeneity of composition, can be combined. Furthermore, we hope to increase the understanding of new constructions by characterising them in terms of more familiar ones. Finally, our overall aim is to apply algebraic methods of specification and verification to the synthesis of concurrent systems; the algebraic treatment of concurrent composition plays an important part in this, since it leads to the possibility of some rather general program transformation strategies, enabling us to bridge the gap between purely declarative axiomatic descriptions of systems behavior, and ``implementation-biased'' descriptions in terms of concurrent processes, without leaving the algebraic framework. We shall have a little more to say on this in the conclusions.


The particular format of our (universal-) algebraic treatment derives from the basic construction of Landin's ``Program-Machine Symmetric Automata Theory'' \cite{L70}, namely algebraic closure in a direct product of partial algebras. We believe that the ideas in the ``Program-Machine Symmetric Automata Theory'' have a potential which has not yet been exploited. Perhaps the present paper will provide some evidence to support this view.

The scheme of the rest of this paper is as follows: section 2 establishes some notation and definitions; section 3 contains the characterisation results; section 4 extends the algebraic constructions in a categorical framework; while the final section summarises what has been done, and suggests some possible extensions and applications.

\section{Preliminaries}

Throughout what follows, $\Sigma$ will be a fixed (not necessarily finite) alphabet. As usual, $\Sigma^*$ denotes the set of all finite strings over $\Sigma$, including the empty string, $\epsilon$. $\Sigma^+ = \Sigma^* - \{\epsilon\}$. We use $s,t,u, \ldots$ to range over $\Sigma^*$ and $\sigma, \sigma_1, \ldots$ to range over $\Sigma$. Thus $st$ means the concatenation of strings $s$ and $t$, while $s\sigma$ means the string $s$ concatenated with (the unit string of) the symbol $\sigma$. $\LL(\Sigma)$, the languages over $\Sigma$, is the set of all subsets of $\Sigma^*$.

\begin{definition} We say that $s$ is a \emph{prefix} of $t$ if, for some $u$, $t = su$. A language $L \in \LL(\Sigma)$ is \emph{prefix-closed} if, for every $s \in L$, and $t$ which is a prefix of $s$, $t \in L$. We are interested in prefix-closed languages as representations of the \emph{behaviours} of systems. The symbols of $\Sigma$ are conceived as names of the various types of \emph{event}, \emph{synchronisation} or \emph{operation} that the system can engage in or perform. A \emph{string} or \emph{trace} or \emph{firing sequence} over $\Sigma$ is a record of occurrences of such events, ordered in time. The (potential) behavior of a system is then characterised by the set of possible traces for the system. Clearly, any prefix of a possible trace is also possible.
\end{definition}

\begin{definition} $\PCL(\Sigma)$, the \emph{prefix closed languages over $\Sigma$}, is the set of all the pairs of $(L, \aL)$ such that

\begin{enumerate}[(i)]
    \item $\aL \subseteq \Sigma$  
    \item $L \subseteq \aL^*$, and $L \neq \empset$
    \item $L$ is prefix closed.
\end{enumerate}
\end{definition}


The component $\aL$ is the (sub-) \emph{alphabet} of the language. $\Sigma$ is the ``global'' alphabet, of all events possible throughout the system. A subsystem may only see or participate in a subset of these events; $\aL$ is meant to distinguish explicitly between those events for which it may exercise some constraint, and those of which it simply does not speak at all.

We shall often employ a mild abuse of notation, speaking of $L \in \PCL(\Sigma)$, and then using $\aL$ when required.

\begin{definition} The \emph{projection} of $\sigma$ onto $\aL \subseteq \Sigma$ is given by
\[ \sigma/\aL = \left\{ \begin{array}{lr}
\sigma & \mbox{if $\sigma \in \aL$} \\
\epsilon & \mbox{otherwise}
\end{array}
\right.
\]


\end{definition}

\noindent Projection extends immediately to strings: $(st)/\aL = (s/\aL)(t/\aL)$.

\begin{fact} $(s/\alpha{L_1})/\alpha{L_2} = (s/\alpha{L_2})/\alpha{L_1} = s/(\alpha{L_1}\cap\alpha{L_2})$
\end{fact}

\noindent We can now define the Campbell-Lauer-Hoare parallel composition operator.

\begin{definition} The operator $\cdot || \cdot : \PCL(\Sigma)^2 \to \PCL(\Sigma)$ is defined by
\[ \begin{array}{lcl}
    \alpha(L_1 || L_2) & = &  \alpha{L_1} \cup \alpha{L_2} \\
    L_1 || L_2 & = &  \{s \in \alpha(L_1 || L_2)* : s/\alpha{L_1} \in L_1 \And s/\alpha{L_2} \in L_2\}.
\end{array}
\]
\end{definition}
    
\noindent This operator is associative and commutative, and we write expressions such as $L_1 || \cdots || L_n$ to indicate n-ary parallel composition.

\begin{definition} For $\uL = (L_1, \ldots ,L_n)$, where $(L_1,\ldots ,L_n) \in \PCL(\Sigma)^n$, $\Vops(L)$, the set of \emph{vector operations} over $\uL$ is:
\[ \Vops(\uL) = \{\myvec{\sigma} : \sigma \in \bigcup_{1 \le i \le n}\alpha{L_i}\} \]
where $\myvec{\sigma} = (\sigma/\alpha{L_1}, \ldots ,\sigma/\alpha{L_n})$.
\end{definition}

\noindent Concatenation is extended elementwise to vectors of strings, e.g.
\[ (s_1, \ldots ,s_n)(t_1, \ldots ,t_n) = (s_1t_1, \ldots ,s_nt_n). \]
This concatenation is clearly associative, and has $\myvec{\epsilon} = (\epsilon, \ldots ,\epsilon)$ as the identity.


Now $\Vops(\uL)^*$ is the monoid generated from $\Vops(\uL)$, i.e. all products of the form
$\myvec{\sigma}_1 \cdots \myvec{\sigma}_n$ for $n \ge 0$,
where $\myvec{\sigma}_1 \in \Vops(\uL), i = 1, \ldots , n$,
and the zero product is $\myvec{\epsilon}$.

This monoid is not free; there are commutation relations $\myvec{\sigma}_1
\myvec{\sigma}_2 = \myvec{\sigma}_2 \myvec{\sigma}_1$, $\sigma_1 \ne 
\sigma_2$, where $\sigma_1$ and $\sigma_2$ are each non-$\epsilon$ only at components 
where the other is $\epsilon$. These commutation relations correspond directly to 
potential concurrency between $\sigma_1$ and $\sigma_2$.

The prefix relation can also be extended elementwise to vectors of strings, i.e. $(s_1, \ldots ,s_n)$ is a prefix of $(t_1, \ldots ,t_n)$ iff $s_i$ is a prefix of $t_i$, $i = 1, \ldots ,n$. WIth this extension, it is clear that the prefix of a vector is a vector of prefixes of the components, and conversely.

\paragraph{Notation} We shall use underlined names to indicate elements of $\Vops(\uL)^*$, e.g. $\us$, $\ut$; and for indexing these vectors of strings we shall write e.g. $[\us]_i$,
meaning the $i$'th component of  $\us$ (which is a string).

\begin{definition}
We write $\PCVL(\Sigma)$ for the set of \emph{prefix-closed vector languages  $(L, (\alpha{L_1}, \ldots ,\alpha{L_n}))$ over $\Sigma$}, where:
\begin{enumerate}
\item $\alpha L_i \subseteq \Sigma$, $i = 1, \ldots , n$.
\item $L \subseteq \Vops(\alpha{L_1}, \ldots ,\alpha{L_n})^*$, $L \neq \empset$.
\item $L$ is prefix-closed.
\end{enumerate}
\end{definition}

\begin{definition}
The operator
\[ \VFS : \PCL(\Sigma)^+ \to \PCVL(\Sigma) \]
mapping a vector of prefix-closed languages to a prefix-closed vector language, the set of \emph{vector firing sequences}, is defined by:
\[ \VFS(L_1, \ldots ,L_n) = (L, (\alpha{L_1}, \ldots ,\alpha{L_n})) \]
where 
\[ L = \Vops(\alpha{L_1}, \ldots , \alpha{L_n})^* \cap (L_1 \times \cdots \times L_n). \]
For every $L \in \PCVL(\Sigma)$, $L = \VFS(L_1, \ldots ,L_n)$, where
$L_i = \{[\us]_i : \us \in L\}$,  $i = 1, \ldots ,n$.
So $\PCL(\Sigma)$ ``spans'' $\PCVL(\Sigma)$ via the $\VFS$ operation. Conversely, $\PCL(\Sigma)$ can be identified with the one-dimensional languages in $\PCVL(\Sigma)$; however, we prefer to keep them notationally distinct.
\end{definition}

\noindent  We now turn our attention to algebra. We will be interested in \emph{partial algebras} \cite{G79}.

\begin{definition}
The \emph{strong equality} relation $\equiv$ compares definedness on both sides as well as values if defined, i.e. $e1 \equiv e2$ for expressions $e1$ and $e2$ iff either $e1, e2$ are both undefined, or they are both defined and $e1 = e2$.
\end{definition}

A \emph{ranked alphabet} is a pairwise disjoint family of sets $\{\Sigma_n\}_{n \ge 0}$ of \emph{operator symbols,} where $\sigma \in \Sigma_n$ has \emph{arity} $n$; $\Sigma_0$ are the nullary operation symbols, or constants.

We can construe our fixed, unranked alphabet $\Sigma$ as a ranked alphabet as follows:
\[ \Sigma_0 = \{\epsilon\}, \qquad
\Sigma_1 = \Sigma, \qquad
\Sigma_n = \empset, \quad n > 1.
\]

\noindent In what follows, the only ranked alphabet we will use is $\Sigma$, but definitions and results valid for ranked alphabets in general will be stated and proved in full generality; we will sometimes use $\Sigma$ for a ``typical ranked alphabet'' and sometimes for our specific, fixed $\Sigma$; the context will always make clear what is intended.


\begin{definition}
$\mathbf{W_\Sigma}$, the \emph{word algebra} for $\Sigma$, is the smallest set of strings over $\Sigma \cup \{\mathbf{(,)}\}$ such that

\begin{enumerate}[(i)]
    \item $\Sigma_0 \subseteq W_\Sigma$
    \item $t_1, \ldots ,t_n \in W_\Sigma$ and $\sigma \in \Sigma_n \; \Rightarrow\;  \mathbf{(}t_1, \ldots ,t_n\mathbf{)}\sigma \in W_\Sigma$.
\end{enumerate}
\end{definition}

\noindent Note that we write operator application on the right; also in the future the parentheses will be ordinary face.

\begin{remark}
There is a natural identification between $\Sigma^*$ and $W_\Sigma$, given by the bijection
\[ \sigma_1, \ldots ,\sigma_n \leftrightarrow ( \cdots (\epsilon)\sigma_1 \cdots)\sigma_n . \]
\end{remark} 

\begin{definition}
    $\PAlg(\Sigma)$ is the class of all \emph{partial algebras} over the ranked alphabet $\Sigma$, i.e. all $A$ such that

\begin{enumerate}[(i)]
    \item $|A|$, the \emph{carrier} of $A$, is a non-empty set
    \item for each $n \ge 0$, $\sigma \in \Sigma_n, \sigma_A$ is an \emph{n-ary partial operation} on $|A|$, i.e. a partial function $\sigma_A : |A|^n \pfn |A|$.
\end{enumerate}

\noindent For $t \in W_\Sigma$, $A \in \PAlg(\Sigma)$, we define $t^A$, the \emph{evaluation} or \emph{run} of $t$ on $A$, or the \emph{denotation} of $t$ in $A$, as follows:

\begin{enumerate}[(i)]
    \item $t = \sigma \in \Sigma_0 \Imp t^A \equiv \sigma_A$
    \item $t = (t_1 , \ldots , t_n)\sigma \Imp t^A \equiv (t_1^A , \ldots , t_n^A)\sigma_A$.
\end{enumerate}

\noindent Note that strong equality is used. This leads to the \emph{definedness predicate,} $\Def t^A$, for $t \in W_\Sigma$, $A \in \PAlg(\Sigma)$, defined by:

\begin{enumerate}[(i)]
    \item $t = \sigma \in \Sigma_0 \Imp \Def t^A = (\sigma_A \ne \empset)$
    \item $t = (t_1 , \ldots , t_n)\sigma \;\Imp\; [\Def t^A \iff \Def t_1^A \And \cdots \And \Def t_n^A \And (t_1^A, \ldots ,t_n^A) \in \dom \sigma_A]$.
\end{enumerate}
\end{definition}

\begin{definition}
A \emph{homomorphism} of partial algebras $\phi: A \to B$, where  $A,B \in \PAlg(\Sigma)$, is a function
$\phi : |A| \to |B|$ such that
\[ \phi((a_1, \ldots ,a_n)\sigma_A) = (\phi(a_1), \ldots ,\phi(a_n))\sigma_B \]
whenever both sides are defined; moreover, definedness on the left hand side implies definedness on the right.
\end{definition}

A \emph{strong homomorphism} is a homomorphism for which definedness on the right hand side of the above equation implies definedness on the left; i.e.~for which the above equation holds with respect to strong equality.
    
An \emph{isomorphism} is a bijective strong homomorphism. If an isomorphism from $A$ to $B$ exists, we say they are \emph{isomorphic} and write $A \cong B$.
    
\begin{definition}
${\times} : \PAlg(\Sigma)^2 \to \PAlg(\Sigma),$ the \emph{direct product of partial algebras}, is defined by:
    
\begin{enumerate}[(i)]
    \item $|A \times B| = |A| \times |B|$\\
    \item $((a_1,b_1), \ldots ,(a_n,b_n))\sigma_{A \times B} \equiv ((a_1, \ldots ,a_n)\sigma_A,(b_1, \ldots ,b_n)\sigma_B)$.
\end{enumerate}

\end{definition}
Note that definedness in the product requires definedness in \emph{both} the factors. Up to isomorphism, the direct product is associative and commutative, and we write expressions like $A_1 \times \cdots \times A_n$ to indicate generalised, $n$-ary direct products.

\begin{definition}
    $A \in \PAlg(\Sigma)$ is \emph{finitely generated} iff every $a \in |A| = t^A$ for some $t \in W_\Sigma$.
\end{definition}
\paragraph{Remark, added 2009}
This is not the standard notion of finite generation, which would be parameterised by a finite subset of $|A|$. Here we are following Landin \cite{L70}.

\begin{definition}
    $B$ is a \emph{subalgebra} of $A$, for $A,B \in \PAlg(\Sigma)$, iff:
    
    \begin{enumerate}[(i)]
        \item $|B| \subseteq |A|$
        \item $(b_1, \ldots ,b_n)\sigma_B \equiv (b_1, \ldots ,b_n)\sigma_A$ for all $b_1, \ldots ,b_n \in |B|$.
    \end{enumerate}
\end{definition}

\begin{definition}
    for $A \in \PAlg(\Sigma)$, $\Ac(A)$, the \emph{algebraic closure} (on the empty set of generators) in $A$, is given by: $\Ac(A) = \bigcup_{n \ge 0} \Ac_n(A)$, where 
\[ \begin{array}{lcl}
\Ac_0(A) & = & \{\sigma_A : \sigma \in \Sigma_0\} \\
\Ac_{n+1}(A) & = &\Ac_n(A) \; \cup \\
& & \{(a_1, \ldots ,a_m)\sigma_A : (a_1, \ldots ,a_m) \in \Ac_n(A)^m \cap \dom \sigma_A, \;\; \sigma \in \Sigma_m\} .
\end{array}
\]
\end{definition}   


\noindent Note that $\Ac(A)$ can be regarded as the least fixed point of an operator on $2^{|A|}$ which is directed-continuous (because the operations are finitary). Then the $\Ac_n(A)$ correspond to the terms of the usual ``ascending Kleene sequence''. The ``least solution'' in this case is the minimal subalgebra of $A$.

\begin{fact} $A \in \PAlg(\Sigma)$ is finitely generated iff $A = \Ac(A)$.
\end{fact}

\begin{fact} If $A$ is finitely generated, then for any $B$ there is at most one homomorphism $\phi : A \to B$, given by $\phi(t^A) = t^B$.
\end{fact}

\section{Characterization Results}

In order to obtain algebraic characterisations of our operations on languages, we need to have a systematic representation of languages as algebras. This is provided by:

\begin{definition}
$F : \PCL(\Sigma) + \PCVL(\Sigma) \to \PAlg(\Sigma)$

\begin{enumerate}[(i)] 
    \item For $L \in \PCL(\Sigma), F(L) = A$, where
\[ \begin{array}{lcl}
|A| & = &  L\\
\epsilon_A & = & \epsilon\\
\sigma_A  & = & \left\{ \begin{array}{lr}
\{ (s, s\sigma) : s\sigma \in L \}, & \sigma \in \alpha L\\
\id_L & \mbox{otherwise}
\end{array} \right.
\end{array}
\]    
    \item For $L = \VFS(L_1, \ldots ,L_n), F(L) = A$, where
\[ \begin{array}{lcl}
|A| & = &  L\\
\epsilon_A & = & \myvec{\epsilon}\\
\sigma_A  & = & \left\{ \begin{array}{lr}
\{ (\us, \us\myvec{\sigma}) :   \us\myvec{\sigma} \in L \}, & \sigma \in \bigcup_i \alpha L_i \\
\id_L & \mbox{otherwise}.
\end{array} \right.
\end{array}
\]    
\end{enumerate}
\end{definition}

\noindent We can now state our characterisation result:

\begin{theorem} 
\label{charth}
For any $L_1, \ldots ,L_n \in \PCL(\Sigma)$:
\[ F(\VFS(L_1, \ldots ,L_n)) = \Ac(F(L_1) \times \cdots \times F(L_n)) . \]
\end{theorem}

\paragraph{Remark added 2009} Thus we can represent the concurrent composition operation, in its ``true concurrency semantics'', by a purely algebraic construction --- the key one from \cite{L70}.
 
\begin{proof}
Let $\uL = (\alpha{L_1}, \ldots ,\alpha{L_n})$, let $A$ be the algebra denoted by the left hand side of the above equation, and $B$ the algebra denoted by the right hand side. The proof proceeds in a number of steps.


\begin{enumerate}[(i)]
    \item $|A| \subseteq L_1 \times \cdots \times L_n \supseteq |B|$. Immediate.
    \item $\epsilon_A = \myvec{\epsilon} = (\epsilon, \ldots ,\epsilon) = (\epsilon_{F(L_1)}, \ldots ,\epsilon_{F(L_n)}) = \epsilon_B$.
    \item For $\us \in \dom \sigma_B$,\\
    $(\us)\sigma_B = \us \myvec{\sigma}$ where $\myvec{\sigma} \in \Vops(\alpha\uL)$.\\
    In fact, $(\us)\sigma_B = ((s_1)\sigma_{F(L_1)}, \ldots ,(s_n)\sigma_{F(L_n)})$,\\
    where 
\[ (s_i)\sigma_{F(L_i)} = \left\{ \begin{array}{lr} 
s_i \sigma & \mbox{if $\sigma \in \alpha L_i$} \\
s_i = s_i \epsilon & \mbox{otherwise}
\end{array} \right.
\]
    So $(\us)\sigma_B = \us([\myvec{\sigma}]_1, \ldots ,[\myvec{\sigma}]_n) = \us \myvec{\sigma}$.
    \item $|A| = |B|$. By (i), we need only show that for $\us \in L_1 \times \cdots \times L_n$, 
    \[ \us \in \Vops(\uL)^* \IFF \us \in |B| . \]

For $\Rightarrow$, we proceed by induction on the length of products in $\Vops(\uL)^*$, which we write $|s|$.

\noindent \textbf{Basis} $|\us| = 0 \Imp\us = \myvec{\epsilon}$. $\us \in |B|$ by (ii).

\noindent \textbf{Induction Step} 
\[ \begin{array}{lll}
& \us = \ut \myvec{\sigma} & \\
\Rightarrow & \ut \in |A| & \mbox{prefix closure} \\
\Rightarrow & \ut \in |B| & \mbox{inductive hypothesis.} 
\end{array}
\]
    Now 
    \[  \begin{array}{lcll}
    \ut \myvec{\sigma} \in |A| & \Rightarrow & [\ut]_i[\myvec{\sigma}]_i \in L_i = |F(L_i)|, & i=1, \ldots ,n \\
    & \Rightarrow & [\ut]_i \in \dom \sigma_{F(L_i)}, &  i = 1, \ldots , n\\
    & \Rightarrow & (\ut)\sigma_B = \ut \myvec{\sigma} \;\; \mbox{(by (iii))} \; \in |B| . &
    \end{array}
    \]
    
For $\Leftarrow$, we have
\[ \begin{array}{lcll}
a \in |B| & \Rightarrow & a \in \Ac_m(F(L_1) \times \cdots \times F(L_n)) & \mbox{for some $m$} \\
    & \Rightarrow & a = (\cdots (\epsilon_B)\sigma_{1B} \cdots )\sigma_{mB}, & \sigma_i \in \Sigma,  \;\; i = 1, \ldots , m\\
    & \Rightarrow & a = \myvec{\epsilon} \myvec{\sigma}_1 \cdots \myvec{\sigma}_m & \mbox{by (3)} \\
    & \Rightarrow & a \in \Vops(\uL)^* . &
    \end{array}
    \]
    
    \item For $\us \in \dom \sigma_A \cap \dom \sigma_B$, $(\us)\sigma_A = \us \myvec{\sigma} = (\us)\sigma_B$ by (iii).
    
    \item For $s \in |A| = |B|$:
    \[ \begin{array}{lcll}
    \us \in \dom \sigma_A
            & \IFF & \us \myvec{\sigma} \in |A| & \\
            & \IFF & [\us]_i [\myvec{\sigma}]_i \in L_i &  i=1, \ldots ,n \\
            & \IFF & [\us]_i \in \dom \sigma_{F(L_i)} &  i=1, \ldots ,n\\
           &  \IFF & \us \in \dom \sigma_B . &
           \end{array}
           \]
\end{enumerate}
\end{proof}

\noindent We now provide a representation of algebras as languages.

\begin{definition} $G : \PAlg(\Sigma) \to \PCL(\Sigma)$.

$G(A) = L$, where
 \[ \begin{array}{lcl}
 \aL & = & \{\sigma \in \Sigma : \sigma_A \ne \id_{|\Ac(A)|}\}\\
 L & = & \{t/\aL : \Def t^A\} .
 \end{array}
 \]
\end{definition}

\begin{remark}
In this definition of G, we blur the distinction between $\Sigma^*$ and $W_\Sigma$, appealing to the identification described in section 2.
\end{remark}

\begin{remark}
Before proceeding, we note that it is consistent with our earlier definitions that for some $A \in \PAlg(\Sigma), \epsilon_A = \empset$, i.e. the constant is undefined. This would mean $G(A) = \empset \notin \PCL(\Sigma)$, i.e. $G$ would be partial. We defer discussion of this and some similar problems to the final section, merely noting their existence as they arise, and revising our definitions accordingly. In this case, we simply restrict $\PAlg(\Sigma)$ to those algebras in which $\epsilon$ is defined. This is what we will mean by $\PAlg(\Sigma)$ for the remainder of this paper.
\end{remark}

We now wish to see how well $F$ and $G$ correspond, by composing them.

\begin{theorem}
\label{corrth}
\begin{enumerate}[(i)]
    \item $G(F(L)) = L$ for all $L \in \PCL(\Sigma)$.
    \item There is a unique strong homomorphism $g_A : F(G(A)) \to A$, for all $A \in \PAlg(\Sigma)$.
    \item $G(\Ac(A)) = G(A)$.
    \item $G(A_1 \times \cdots \times A_n) = G(A_1) \, || \, \cdots \, || \, G(A_n)$.
\end{enumerate}
\end{theorem}

\begin{corollary}
\label{corrcorr}
$G(F(\VFS(L_1, \ldots ,L_n))) = L_1 \, || \, \cdots \, || \,  L_n$.
\end{corollary}

\noindent We shall firstly prove a few lemmas before establishing Theorem~\ref{corrth}.

\begin{fact}$\Def t^{F(L)} \; \Rightarrow \; t^{F(L)} = t/\aL$.
\end{fact}

\begin{fact} $\Def t^{F(L)} \IFF \Def (t/\aL)^{F(L)}$.
\end{fact}

\begin{lemma}
\label{lemm1}
\[ t/\aL \in L \IFF [\Def t^{F(L)} \And t^{F(L)} = t/\aL] . \]
\end{lemma}

\begin{proof}
By induction on $|t|$.
    
\noindent \textbf{Basis} $t = \epsilon = \epsilon/\aL \in L$ iff $\Def \epsilon^{F(L)}$ and $\epsilon^{F(L)} = \epsilon$.

\noindent \textbf{Induction Step} $t = s\sigma$.
Firstly, $\Rightarrow$. Assume $t/\aL \in L$.
\begin{description}
\item[Case (i)] $\sigma \in \aL$:
\[ \begin{array}{llll}
        & & t/\aL = (s/\aL)\sigma \in L & \\
        & \Rightarrow & (s/\aL) \in L &           \mbox{prefix closure} \\
    (*) & \Rightarrow & \Def s^{F(L)} \And s^{F(L)} = s/\aL &  \mbox{inductive hypothesis} 
    \end{array}
     \]
 \[ \begin{array}{lcll}       
 (s/\aL)\sigma \in L & \Rightarrow & (s/\aL, (s/\aL)\sigma) \in \sigma_{F(L)} &   \mbox{by definition of $F$}\\
                & \Rightarrow & \Def t^{F(L)} \And t^{F(L)} = (s/\aL\sigma) = t/\aL,  & \mbox{using (*)}.
 \end{array}
 \]
                
\item[Case (ii)] $\sigma \notin \aL$.
\[ t/\aL = s/\aL \in L \Imp \Def s^{F(L)} \And s^{F(L)} = s/\aL .\]
\[ \begin{array}{l}
        \sigma_{F(L)} = id_L \quad \mbox{by definition of $F$} \\
        \Rightarrow (s/\aL, s/\aL) \in \sigma_{F(L)}  \\
        \Rightarrow \Def t^{F(L)} \And t^{F(L)} = s/\aL = t/\aL. 
        \end{array}
        \]
\end{description}

\noindent Secondly, $\Leftarrow$. Assume $\Def t^{F(L)} \And t^{F(L)} = t/\aL$. Then 
\[ t^{F(L)} = t/\aL \in |F(L)| = L . \]
\end{proof}

\begin{lemma}
\label{lemm2}
$t^{(A_1 \times \cdots \times A_n)} \equiv (t^{A_1}, \ldots ,t^{A_n})$.
\end{lemma}

\begin{lemma}
\label{lemm3}
\begin{enumerate}[(i)]
    \item $\sigma_{F(L)} = id_{F(L)} \IFF \sigma \notin \aL$
    \item $\alpha{G}(\Ac(A)) = \alpha{G(A)}$
    \item $\alpha{G}(A_1 \times \cdots \times A_n) = \bigcup_i\alpha{G(A_i)}$.
\end{enumerate}
\end{lemma}

\begin{proof}
\begin{enumerate}[(i)]
    \item Immediate, but needs that $L \ne \empset$.
    \item Immediate, but needs that $\alpha{G(A)}$ is defined via $\Ac(A)$ rather than $A$.
    \item Immediate from Lemma~\ref{lemm2}.
\end{enumerate}
\end{proof}

We now prove Theorem~\ref{corrth}.
\begin{proof} 
    \begin{enumerate}[(i)]
    \item 
    \[ \begin{array}{cll}
    & s \in L & \\
    \Rightarrow & \Def s^{F(L)} & \mbox{Lemma~\ref{lemm2}}    \\
    \Rightarrow & s/\alpha{G(F(L))} = s/\aL & \mbox{Lemma~\ref{lemm3}(i)}  \\
                    = & s \in G(F(L)) &           \\
     \Rightarrow & \exists t. \, s = t/\aL \And \Def t^{F(L)} & \\
      \Rightarrow & t/\aL = s \in L. &
      \end{array}
      \]
                        
    \item Define $g_A : F(G(A)) \to A$ by: $g_A(t) = t^A$. Then:
    \[ \begin{array}{lll}
        & \Def t^{F(G(A))}\\
        \IFF & t/\alpha{G(A)} \in G(A) & \mbox{by Lemma~\ref{lemm1}} \\
         \IFF & \Def t^A & \mbox{by definition of G.} 
         \end{array}
         \]
         Since $F(G(A))$ is finitely generated, and does not identify terms,
          this proves $g_A$ is the unique strong homomorphism.
          
    \item $G(\Ac(A)) = G(A)$. Immediate.
    
    \item $\alpha{G(A_1 \times \cdots \times A_n)} =
        \bigcup_\alpha{G(A_i)} =
        \alpha(G(A_1) \, || \cdots || \, G(A_n))$.\\
        Call this alphabet $\aL$, then for $s \in \alpha{L^*}$:
        \[ \begin{array}{lll}
        & s \in G(A_1 \times \cdots \times A_n) & \\
         \IFF & \Def s^{A_1 \times \cdots \times A_n} & \\
          \IFF & \Def s^{A_i}$,   $i = 1, \ldots ,n & \mbox{by Lemma~\ref{lemm2}} \\
           \IFF & s/\alpha{G(A_i)} \in G(A_i), \quad i=1, \ldots ,n & \\
            \IFF & s \in G(A_1) \, || \cdots || \, G(A_n).
            \end{array}
            \]
\end{enumerate}
\end{proof}

We now prove Corollary~\ref{corrcorr}.

\begin{proof} 
\[ \begin{array}{lcll}
& & G(F(\VFS(L_1, \ldots ,L_n))) & \\
       & = & G(\Ac(F(L_1) \times \cdots \times F(L_n)))    &  \mbox{by Theorem~\ref{charth}}    \\
        & = & G(F(L_1) \times \cdots \times F(L_n))        & \mbox{by Theorem~\ref{corrth}(iii)} \\
        & = & G(F(L_1)) \, || \cdots || \, G(F(L_n))          &     \mbox{by Theorem~\ref{corrth}(iv)} \\
        & = & L_1 \, || \cdots || \, L_n                              & \mbox{by Theorem~\ref{corrth}(i)}.  
 \end{array}
 \]
\end{proof}

\noindent $G$ gives a representation of algebras in $\PCL(\Sigma)$; we now consider how to find a representation in $\PCVL(\Sigma)$. To do this, we need some additional structure on algebras, to indicate how they are composed from sequential components. This is provided by the standard algebraic idea of \emph{sub-direct product decomposition}.

\begin{definition}
A \emph{congruence} $\Theta$ on $A \in \PAlg(\Sigma)$ is an equivalence relation on $|A|$ satisfying the \emph{substitution property}:
\[ a_1 \Theta b_1, \ldots , a_n \Theta b_n \Imp (a_1, \ldots ,a_n)\sigma_A \Theta (b_1, \ldots ,b_n)\sigma_A \]
provided both sides are defined. A \emph{strong congruence} is a congruence with the property that each side in the above expression is defined iff the other is.
\end{definition}

\begin{definition}
The \emph{quotient algebra} of $A \in \PAlg(\Sigma)$ by a congruence $\Theta$, written $A/\Theta$, is defined as follows:
\[ |A/\Theta| = \{[a]_\Theta : a \in |A|\}, \quad   \mbox{the equivalence classes by $\Theta$.} \]
 \[ ([a_1], \ldots ,[a_n])\sigma_{A/\Theta} = [(b_1, \ldots ,b_n)\sigma_A] \]
if there exist some $b_1, \ldots ,b_n$ with $b_i \Theta a_i$, $i=1, \ldots ,n$, such that $b_1, \ldots ,b_n \in \dom \sigma_A$; and otherwise undefined.
\end{definition}

\begin{definition}
A \emph{sub-direct decomposition} of $A \in \PAlg(\Sigma)$ is a list $(\Theta_1, \ldots ,\Theta_n)$ of congruences on $A$ such that $A$ is isomorphic to a subalgebra of $A/\Theta_1 \times \cdots \times A/\Theta_n$.
\end{definition}

\begin{definition}
$\mathbf{\SDPA(\Sigma)}$ is the class of all $(A,(\Theta_1, \ldots ,\Theta_n))$ such that $A \in \PAlg(\Sigma)$ and $(\Theta_1, \ldots ,\Theta_n)$ is a sub-direct decomposition of $A$.
\end{definition}

\begin{remark}
An algebra can have many sub-direct decompositions. $\PAlg(\Sigma)$ is embedded in $\SDPA(\Sigma)$ by
$A \mapsto (A, (id_A))$.
\end{remark}


\noindent We now sharpen our definition of $F$ to respect product structure.

\begin{definition}
We define $F' : \PCVL(\Sigma) \to \SDPA(\Sigma)$:\\
\[ F'(L) = (F(L), (\ker \proj_1, \ldots ,\ker \proj_n)) \]
where $\proj_i$ is the i'th projection function on $L$, and $\ker f$ is the binary relation $f^{-1}f$.
\end{definition}

\begin{lemma}
\label{lemm4}
    For $L = \VFS(L_1, \ldots ,L_n) \in \PCVL(\Sigma)$,\\
    \[ F'(L)/(\ker \proj_i) \equiv F(\proj_i(L),\alpha{L_i}), \quad i = 1, \ldots ,n. \]
\end{lemma}

\begin{proof}
We take $\phi : [\us]_{\ker \proj_i} \to [\us]_i$. Now
\[ \us(\ker \proj_i)\ut \IFF [\us]_i = [\ut]_i \]
so $\phi$ is well defined and injective, and obviously surjective.
Let $A = F'(L)/(\ker \proj_i)$ and $B = F(\proj_i(L))$. Then:
\[ \begin{array}{lcl}
& & \Def [\us]\sigma_A \\
    & \IFF &\ut \myvec{\sigma} \in L \quad \mbox{for some $\ut (\ker \proj_i) \us$} \\
    & \IFF &[\ut \myvec{\sigma}]_i = [\us \myvec{\sigma}]_i \in \proj_i(L) \\
    & \IFF &[\us]_i \in \dom \sigma_B \\
    & \IFF &\Def (\phi([\us]))\sigma_B.
    \end{array}
    \]
    
\noindent Moreover, $\phi(([\us\sigma_A])) = [\us]_i [\myvec{\sigma}]_i = (\phi([\us]))\sigma_B$. Thus, $\phi$ is a strong homomorphism.
\end{proof}

\begin{corollary}
    $(\ker \proj_1, \ldots ,\ker \proj_n)$ is a subdirect decomposition of $F(L)$; i.e. $F'(L)$ is well defined.
\end{corollary}

\begin{proof}
    For $L \in \PCVL(\Sigma)$:
  \[ \begin{array}{lcl}
  F(L) & = & F(\VFS(\proj_1(L), \ldots ,\proj_n(L)))\\
        & = & \Ac(F(\proj_1(L)) \times \cdots \times F(\proj_n(L)))\\
        & = & \Ac(F(L)/(\ker \proj_1) \times \cdots \times F(L)/(\ker \proj_n)).
 \end{array}
 \]
\end{proof}

We now define our vector analogue of $G$.
\begin{definition}
    $H : \SDPA(\Sigma) \to \PCVL(\Sigma)$.
    
    $H(A, (\Theta_1, \ldots ,\Theta_n)) = L$, where
    \[ \begin{array}{lcl}
       \aL & = & (G(A/\Theta_1), \ldots ,G(A/\Theta_n))\\
       L & = & \{\ut \in \Vops(\aL)^* : \Def t^A\} .
       \end{array}
       \]
\end{definition}

\begin{theorem}
\label{veccorrth}
\begin{enumerate}[(i)]
    \item $H(F(\VFS(L_1, \ldots ,L_n))) = \VFS(L_1, \ldots ,L_n)$
    
    \item For $A \in \SDPA(\Sigma)$, there is a unique strong homomorphism
  \[ h_A : F'(H(A)) \to A \]
        
    \item $H(\Ac(A)) = H(A)$
    
    \item $H(A, (\Theta_1, \ldots ,\Theta_n)) = \VFS(G(A/\Theta_1), \ldots ,G(A/\Theta_n))$.
\end{enumerate}
\end{theorem}

\begin{lemma}
\label{lemm5}
    For $L = \VFS(L_1, \ldots ,L_n) \in \PCVL(\Sigma)$: 
\[ \alpha{G(F(L)/\ker \proj_i)} = \alpha{L_i}, \quad i=1, \ldots ,n . \]
\end{lemma}

\begin{proof}
    By Lemma~\ref{lemm4}, $\alpha{G(F(L)/\ker \proj_i)} = \alpha{G(F(\proj_i(L)))} = \alpha{L_i}$.
\end{proof}

\begin{lemma}
\label{lemm6}
For any $\ut = \myvec{\sigma}_1 \cdots \myvec{\sigma}_m \in \Vops(L_1, \ldots ,L_n)^*$:
\[ [\ut]_i = (\sigma_1 \cdots \sigma_m)/\alpha{L_i}, \quad i=1, \ldots ,n . \]
\end{lemma}

\begin{remark}
    This lemma shows that $L_1 \, || \cdots || \, L_n$ and $\VFS(L_1, \ldots ,L_n)$ comprise \emph{exactly} the same set of ``formal products'' of symbols in $\Sigma$; the difference being that for $L_1\, || \cdots ||\, L_n$, the products are interpreted freely, in $\Sigma^*$, while for $\VFS(L_1, \ldots ,L_n)$ they are interpreted in the non-free monoid $\Vops(\alpha{L_1}, \ldots ,\alpha{L_n})^*$.
\end{remark}

We now prove Theorem~\ref{veccorrth}.
\begin{proof} 
\begin{enumerate}[(i)]
    \item Firstly, $\Vops(H(F'(L))) = \Vops(L)$, by Lemma~\ref{lemm5}.
        Now for $\ut \in \Vops(L)^*$:
         \[ \begin{array}{lcll}
    & & \ut \in H(F'(L)) \\
            & \IFF & \Def t^{F'(L)} \equiv t^{F(L)} & \\
            & \IFF &\Def t^{\Ac(F(L_1) \times \cdots \times F(L_n))} & \mbox{Theorem~\ref{charth}} \\
            & \IFF &\Def t^{F(L_1) \times \cdots \times F(L_n)} & \\
            & \IFF &\Def t^{F(L_i)}, \quad i=1, \ldots ,n & \mbox{Lemma~\ref{lemm2}}\\
            & \IFF &t/\alpha{L_i} \in L_i, \quad i=1, \ldots ,n & \mbox{Lemma~\ref{lemm1}} \\
            & \IFF &[\ut]_i \in L_i, \quad i=1, \ldots ,n &   \mbox{Lemma~\ref{lemm6}} \\
            & \IFF & \ut \in L. &
            \end{array}
            \]
            
    \item As for Theorem~\ref{corrth}.
    
    \item Immediate.
    
    \item Call the left hand side $L$, the right hand side $L'$. $\Vops(L) = \Vops(L')$,
    by definition. Now, for $t \in \Vops(L)^*$:
    \[ \begin{array}{lcll}
    & & t \in L & \\
       &  \IFF & \Def t^A & \\
        & \IFF & \Def t^{A/\Theta_1 \times \cdots \times A/\Theta_n} & \mbox{$\Ac$ minimal} \\
        & \IFF &  \Def t^{A/\Theta_i},        \quad i=1, \ldots ,n       &  \mbox{Lemma~\ref{lemm2}} \\
        & \IFF & \Def t^{F(G(A/\Theta_i))}, \quad i=1, \ldots ,n    &  \mbox{Theorem~\ref{corrth}(ii)}\\
        & \IFF & t/\alpha{G(A/\Theta_i)} \in G(A/\Theta_i), \quad i=1, \ldots ,n &  \mbox{Lemma~\ref{lemm1}} \\
        & \IFF & [\ut]_i \in G(A/\Theta_i), \quad i=1, \ldots ,n & \mbox{Lemma~\ref{lemm6}} \\
        & \IFF & \ut \in G(A/\Theta_1) \times \cdots \times G(A/\Theta_n) & \\
        & \IFF & \ut \in L'. & 
        \end{array}
        \]
\end{enumerate}
\end{proof}

\section{Categorical Structure}

In this section, we introduce some more structure. We define morphisms in $\PCL(\Sigma)$, $\PCVL(\Sigma)$, $\PAlg(\Sigma)$ and $\SDPA(\Sigma)$, which become categories; and we extend $F$, $F'$, and $G$ to act on morphisms; they become functors. We re-examine the compositions of $F$ and $G$, and find that the same properties persist that were expressed in Theorem 2. In addition, we find that $F$ is left adjoint to $G$. Much of the development is reminiscent of categorical automata theory (see e.g. \cite{Gog72,AM74})\footnote{These references added in 2009.}. Thus, $F$ takes languages to a form of ``free realisation''; $G$ takes algebras, thought of as acceptors for which every element is a final state, to their behaviours,
i.e.~to the prefix-closed languages which they accept. Our choice of morphisms may be less standard. For a definition of ``simulations'' between languages thought of as behaviours, it turns out that the corresponding algebraic notion is that of \emph{derivor}, and \emph{derived homomorphism}.

\begin{definition}A \emph{simulation} of $L$ by $L'$, for $L,L' \in \PCL(\Sigma)$, is a function
\[ f : \aL \to (\aL')^* \]
extended to $L^*$ by
\[ f(st) = f(s)f(t) \]
which has the (weak) \emph{simulation property}:
\[ s \in L \Imp f(s) \in L' . \]
A \emph{strong simulation} of $L$ by $L'$ is a simulation $f : L \to L'$ with the additional \emph{strong simulation property}:
\[ f(s) \in L' \Imp s \in L . \]
\end{definition}

\begin{remark}
Intuitively, $f$ represents each atomic behavior of $L$ by a (possibly compound) behavior of $L'$; the simulation property ensures that every possible behaviour of $L$ is represented by a possible behaviour of $L'$; the strong simulation property provides a converse.
\end{remark}

\begin{fact}
Simulations and strong simulations are closed under composition of set functions.
\end{fact}

\begin{fact}
For $L \in \PCL(\Sigma)$, $\id_L$ is a strong simulation.
\end{fact}

It follows from the above facts that we can make the following definition.

\begin{definition}
$\CPCL$ is the category given by the following data:

\begin{center}
\begin{tabular}{ll}
    Objects:      &  $\PCL(\Sigma)$\\
    Morphisms:    &  Simulations\\
    Identity:     &  $\mathbf{1}_L = \id_L, L \in \PCL(\Sigma)$\\
    Composition:  &  inherited from \textbf{Set}\\
    \end{tabular}
    \end{center}
    
Moreover, we obtain a subcategory by restricting to strong simulations.
\end{definition}

\noindent We now consider how to extend these ideas to $\PCVL(\Sigma)$. Since $\Vops(\aL)^*$ is not free, we cannot always extend a map on the generators (i.e.~$\aL$) to a monoid homomorphism. However, the commutation relations are all that has to be checked for. This is made precise by the following Lemma.

\begin{lemma}
    For any identity in $\Vops(\uL)^*$
\[ \us = \us_1 \cdots \us_n = \ut_1 \cdots \ut_m = \ut \]
we have:
\begin{enumerate}[(i)]
    \item $n = m$
    
    \item each side can be obtained from the other as a composition of commutation relations in $\Vops(\uL)$.
\end{enumerate}
\end{lemma}

\begin{proof}
    For each $\ua \in \Vops(\uL)$, the number of occurrences of $\ua$ in $s$ = number of occurrences in $\ut$; for, choose $i$ such that $[\ua]_i \ne \epsilon$. (We can always do this by definition of $\Vops(\uL)$). Then by $[\us] = [\ut]$, $[\us]_i = [\ut]_i$ as strings, and so certainly as multisets. So $m = n$, and $\ut = \us_{i_1} \cdots \us_{i_n}$, where $i$ is a permutation of $\{1, \ldots ,n\}$.
    
    Clearly it is sufficient to suppose that $\us_1 \ne \us_{i_1}$, and show that, if 
$j$ is the least number such that $\us_{i_j} = \us_1$, $\us_{i_j}$ commutes with 
$\us_{i_k}$ for all $k<j$. So suppose not, and let $k$ be the least index $<j$ such that 
$\us_{i_j}\us_{i_k} \ne \us_{i_k}\us_{i_j}$. Choose $l$ such that $[\us_{i_k}]_l \ne 
\epsilon \ne [\us_{i_j}]_l$. Then $[\us]_l$ has $s_1$ as its first element, while 
$[\ut]_l$ has $\{s_{i_p} : p$ is the least index $\le k < j$ such that $[\us_{i_p}]_l 
\ne \epsilon\}$ as its first element. This set is non-empty, since the set it minimises 
contains at least $k$, by the supposition that $k$ exists. We now have a contradiction, 
since $\us = \ut$, but $[\us]_l(1) = s_1 \ne s_{i_p} = [\ut]_l(1)$.
\end{proof}

\begin{corollary}
Every function
\[ f : \Vops(L) \to \Vops(L')^*, \qquad L,L' \in \PCVL(\Sigma) \]
which satisfies
\[ \ua \ub = \ub \ua \Imp f(\ua)f(\ub) = f(\ub)f(\ua), \qquad \ua,\ub \in \Vops(\uL) \]     
can be uniquely extended to a monoid homomorphism
\[ f : \Vops(L)^* \to \Vops(L')^* . \]
Such a function is said to \emph{preserve concurrency.}
\end{corollary}

\begin{definition}
A \emph{simulation} of $L$ by $L'$, $L,L' \in \PCVL(\Sigma)$, is a concurrency-preserving function
\[ f : \Vops(L) \to \Vops(L')^* \]
extended to $\Vops(L)^*$ by
\[ f(\us \ut) = f(\us)f(\ut) \]
which has the \emph{simulation property}:
\[ \us \in L \Imp f(\us) \in L' . \]
A \emph{strong simulation} of $L$ by $L'$ is a simulation $f:L \to L'$ with the additional \emph{strong simulation property}:
\[ f(\us) \in L' \Imp \us \in L . \]
\end{definition}

\begin{fact}
    Simulations and strong simulations are closed under composition as set functions.
\end{fact}

\begin{fact}
    The identity function $\id_L, L \in \PCVL(\Sigma)$, is a strong simulation.
\end{fact}

It follows from the above facts that we can make the following definition.

\begin{definition}
    $\CPCVL$  is the category defined as follows:
    
\begin{center}
\begin{tabular}{ll}
    Objects:      &  $\PCVL(\Sigma)$\\
    Morphisms:    &  Simulations\\
    Identity:     &  $\mathbf{1}_L = \id_L, L \in \PCVL(\Sigma)$\\
    Composition:  &  inherited from \textbf{Set}\\
    \end{tabular}
    \end{center}
    
\noindent Moreover, we obtain a subcategory of $\CPCVL$  by restricting to strong simulations. Finally, if we restrict to one-dimensional vector languages in $\CPCVL$ , we obtain $\CPCL$ as a full subcategory, under the natural identification.
\end{definition}

We now turn our attention to $\PAlg(\Sigma)$. To capture the idea of a simulation, we want a way of systematically construing the basic operations of one algebra in terms of (possibly ``compound'', i.e. derived) operations of another. But this is just the notion of a \emph{derivor}, defined by ADJ in their work on algebraic treatments of data abstractions \cite{ADJ78}, and itself essentially equivalent to the categorical idea of \emph{theory morphism}.

\begin{definition}
    We take $X = \{x_0, x_1, x_2, \ldots \}$ to be a set of ``variables'' disjoint with $\Sigma$. $X_n = \{x_0, \ldots , x_{n-1} \}$. For a ranked alphabet $\Sigma$, $W_\Sigma(X)$ is the algebra of terms with variables: it is equal to $W_{\Sigma'}$, where
\[ \begin{array}{lcl}
\Sigma'_0 &  = &  \Sigma_0 \cup X\\
\Sigma'_n & = & \Sigma_n, \quad n > 0 .
\end{array}
\]
The algebra of n-ary terms, or terms in n variables, $W_\Sigma(X_n)$, is defined similarly.
\end{definition}

\begin{definition}
For $t \in W_\Sigma(X_n), t_1, \ldots ,t_n \in W_\Sigma(X)$, \emph{simultaneous substitution} $t[t_1, \ldots ,t_n]$ is the result of replacing all occurrences in $t$ of $x_i$ by $t_i$, $i=1, \ldots ,n$. For a more formal definition, and a proof that substitution is associative, cf. \cite{ADJ77}.
    
Furthermore,
\[ t^A[a_1, \ldots ,a_n], \qquad t \in W_\Sigma(X_n), A \in \PAlg(\Sigma), a_1, \ldots ,a_n \in |A| \]
means the following:
extend $\Sigma$ to $\Sigma'$ by
\[ \begin{array}{lcl}
\Sigma'_0 &  = &  \Sigma_0 \cup \{a_1, \ldots ,a_n\} \\
\Sigma'_n & = & \Sigma_n, \quad n > 0 .
\end{array}
\]
and extend $A$ to a $\Sigma'$ algebra $A'$, by setting $a_{i}^{A} = a_i$, $i=1, \ldots ,n$.
Then 
\[ t^A[a_1, \ldots ,a_n] = (t[a_1, \ldots ,a_n])^{A'} . \]
\end{definition}

\begin{definition}
For ranked alphabets $\Sigma$, $\Sigma'$ a \emph{derivor} from $\Sigma$ to $\Sigma'$ is a family of functions $d = \{d_n\}_{n\ge0}$, where
\[ d_n : \Sigma_n \to W_{\Sigma'}(X_n) . \]
A derivor is extended to act on $W_\Sigma(X)$ by
\[ \begin{array}{lcll}
    d(x) & = &  x, & x \in X\\
    d(\sigma) & = & d_0(\sigma), & \sigma \in \Sigma_0\\
    d((t_1, \ldots ,t_n)\sigma) & = &  d_n(\sigma)[d(t_1), \ldots ,d(t_n)], & \sigma \in \Sigma_n .
    \end{array}
    \]
We shall only have occasion to consider derivors from $\Sigma$ to $\Sigma$.
\end{definition}

\begin{definition}
For $d: \Sigma \to \Sigma$ a derivor and $A \in \PAlg(\Sigma)$, $dA$, the \emph{derived algebra} from $A$ by $d$ is defined as follows:
\[ \begin{array}{lcl}
|dA| & = &  |A|\\                      
\sigma_{dA} & = &  \lambda(a_1, \ldots ,a_n) . \, (d\sigma)^A[a_1, \ldots ,a_n], \qquad \sigma \in \Sigma_n .
\end{array}
\]
\end{definition}

\begin{definition}
A \emph{derived homomorphism} from $A$ to $B$, $A,B \in \PAlg(\Sigma)$, is a pair $(d,\phi)$ such that:
    
\begin{enumerate}[(i)]
    \item $d$ is a derivor from $\Sigma$ to $\Sigma$.
    \item $\phi$ is a homomorphism from $A$ to $dB$.
\end{enumerate}

\noindent A \emph{strong} derived homomorphism is a derived homomorphism $(d,\phi)$ such that $\phi$ is a strong homomorphism.
\end{definition}

\begin{definition}
The \emph{identity derivor} $\Id$ is defined by
\[ \Id_n : \sigma \mapsto (x_0, \ldots ,x_{n-1})\sigma, \qquad n \ge 0 . \]
\end{definition}

\begin{fact}
$\Id \, t = t$ for $t \in W_\Sigma(X)$.
\end{fact}

\noindent Note that homomorphisms can be embedded in derived homomorphisms by
\[ \phi \mapsto (\Id,\phi) . \]
    
\begin{lemma}
\label{lemm7}
For $A,B \in \PAlg(\Sigma)$, $\phi : A \to B$ a homomorphism, $t \in W_\Sigma(X_n)$, and $a_1, \ldots ,a_n \in |A|$:
\[ \phi(t^A[a_1, \ldots ,a_n]) = t^B[\phi(a_1), \ldots ,\phi(a_n)] \]
where definedness on the left hand side implies definedness on the right. If $\phi$ is strong, the equation holds with strong equality.
\end{lemma}

\begin{lemma}
\label{lemm8}
    For $d$ a derivor, and other notation as in Lemma~\ref{lemm7},
\[ t^{dA}[a_1, \ldots ,a_n] \equiv dt^A[a_1, \ldots ,a_n] . \]
\end{lemma}

\begin{lemma}
For $(d,\phi) : A \to B$, $(d', \phi'): B \to C$ derived homomorphisms:
\[ \begin{array}{cll}
    & \phi' \circ \phi((a_1, \ldots ,a_n)\sigma_A) & \\
    = &  \phi'(\phi((a_1, \ldots ,a_n)\sigma_A)) & \\
    = &  \phi'(d\sigma^B[\phi(a_1), \ldots ,\phi(a_n)]) & \mbox{$(d,\phi)$ a derived hom} \\
    = &  (d\sigma)^{d'C}[\phi' \circ \phi(a_1), \ldots ,\phi' \circ \phi(a_n)] & \mbox{Lemma~\ref{lemm7}} \\
    = &  d'(d\sigma)^C[\phi' \circ \phi(a_1), \ldots ,\phi' \circ \phi(a_n)] & \mbox{Lemma~\ref{lemm8}} \\
    = &  (d'd)\sigma^C[\phi' \circ \phi(a_1), \ldots ,\phi' \circ \phi(a_n)] & \mbox{associativity of substitution} \\
    = &  (\phi' \circ \phi(a_1), \ldots ,\phi' \circ \phi(a_n))\sigma_{d'dC} & \mbox{definition of derived algebra}.
    \end{array}
    \]
\end{lemma}

\begin{corollary}
    Derived homomorphisms are closed under elementwise functional composition.
\end{corollary}

\begin{remark}
    Derivors are too general to match with simulations as they stand. We need the following restrictions on derivors $d$:
\[ \begin{array}{ll}
(D1) & d \epsilon = \epsilon \\
(D2) & d \sigma \in W_\Sigma(X) - W_\Sigma \\
\end{array}
\]
and on  derived homomorphisms $(d,\phi) : A \to B$:
 \[ \begin{array}{lll}
(D3) & d \sigma \in W_{\alpha G(B)}(X) & \sigma \in \alpha G(A)\\
          & d \sigma = x_0 & \sigma \notin \alpha G(A).
    \end{array}
    \]
    
\noindent We shall call derivors and derived homomorphisms satisfying the above restrictions \emph{canonical}. For an explanation, see Section~5.
\end{remark}

\begin{definition}
For $A \in \PAlg(\Sigma)$, $\Id_A$ is the canonical derivor defined by:
\[ 
\Id_A \, \sigma =  \left\{ \begin{array}{ll}
   (x_0)\sigma & \sigma \in \alpha G(A)\\
     x_0 &             \sigma \notin \alpha G(A).
      \end{array} \right.
\]
\end{definition}

\begin{fact}
    $\Id_A$ is the identity for composable canonical derivors.
\end{fact}

From the above results, it follows that we can make the following definitions.

\begin{definition}
The category $\CPAlg$ is given by

\begin{center}
\begin{tabular}{ll}
    Objects:      &  $\PAlg(\Sigma)$\\
    Morphisms:   &   Canonical derived homomorphisms\\
    Identity:    &   $\mathbf{1}_A = (\Id_A, \id_{|A|})$,       $A \in \PAlg(\Sigma)$\\
    Composition:   &  elementwise composition inherited from \textbf{Set}.
    \end{tabular}
    \end{center}

\noindent Moreover, we can form a sub-category by restricting to strong derived homomorphisms.
\end{definition}

\begin{definition}
The category $\CSDPA$ is defined as follows:
    
    \begin{center}
    \begin{tabular}{ll}
    Objects:     &   $\SDPA(\Sigma)$\\
    Morphisms:   &   $\hom((A,\myvec{\Theta}),(B,\myvec{\Theta}')) = \hom(A,B)$ in $\CPAlg$\\
    Identity, Composition: & as in $\CPAlg$.
    \end{tabular}
    \end{center}
\end{definition}

\begin{definition}
$\UU : \CSDPA \to \CPAlg$  is a forgetful functor:
\[ \begin{array}{lcl}
\UU(A,\myvec{\Theta}) & = &  A \\
\UU(d,\phi) & = & (d,\phi) .
\end{array}
\]
\end{definition}

We now show how to extend $F$, $F'$, and $G$ to act on morphisms.

\begin{definition}
    $\BB : \CPCL \to \CPAlg$.
    
    \noindent On objects, $\BB$ coincides with $F$. On morphisms $f: L \to L'$,  $\mathbf{\BB}(f) = (d,\phi)$, where:
        \[ \begin{array}{lcll}
        d(\sigma) & = & (\cdots (x_0)\sigma_1\cdots )\sigma_n & \sigma \in \aL, f(\sigma) = \sigma_1 \cdots \sigma_n \\
        d(\sigma) & = & x_0 & \sigma \notin \aL \\
        d(\epsilon) & = & \epsilon. &
        \end{array}
        \]
 Since $F(L)$ is finitely generated, $\phi$ is uniquely determined by $d$.
   \end{definition} 
 
 \begin{definition}
 $\Bi : \CPCVL \to \CSDPA$.
 
\noindent Again, on objects $\Bi$ coincides with $F'$; and $\Bi$ is defined on morphisms exactly like $\BB$.
\end{definition}

\begin{definition}
$\GG : \CPAlg \to \CPCL$

\noindent On objects, $\GG$ coincides with $G$.\\
 Morphisms: for $(d,\phi) : A \to B$,
  \[ \GG (d,\phi) = f \]
 where, for $\sigma \in \alpha G(A)$,  if $d(\sigma) = (\ldots(x)\sigma_1\ldots)\sigma_n$,
 then $f(\sigma) = \sigma_1 \cdots \sigma_n$,
 and $f$ is uniquely determined by its action on $\alpha G(A)$, as in the definition of simulation.

Note that $f$ depends only on $d$, and not at all on $\phi$. This fits with the way that $G$ on objects depends only on the finitely generated portion of an algebra, since $\phi$ is uniquely determined by $d$ over the finitely generated portion.
\end{definition}

\begin{theorem}
\begin{enumerate}[(i)]
\item $\BB$ is a functor, and takes strong simulations to strong derived homomorphisms.
    
\item $\GG$ is a functor, and takes strong derived homomorphisms to strong simulations.

\item $\GG \BB = \mathbf{Id}_{\CPCL}$, the identity functor.

\item For $(d, \phi) : A \to B$, $A, B \in \CPAlg$, the following square commutes:

\begin{diagram}
\BB\GG(A) & \rTo^{\BB\GG(d, \phi)} & \BB\GG(B)) \\
\dTo^{(\Id_{A}, g_{A})} & & \dTo_{(\Id_{A}, g_{A})} \\
A & \rTo_{(d, \phi)} & B
\end{diagram}
This says that the family $\{ (\Id_{A}, g_{A}) \}$ form a natural transformation from $\BB\GG$ to $\mathbf{I}_{\CPAlg}$. Here $g_A$ is the evaluation functional $t \mapsto t^A$ introduced in the proof of Theorem~\ref{corrth}.
\end{enumerate}
\end{theorem}

\begin{theorem}
\begin{enumerate}[(i)]
\item $\Bi$ is a functor, and takes strong simulations to strong derived homomorphisms.
    
\item For $L = \VFS(L_1, \ldots ,L_n)$ in $\PCVL(\Sigma)$, there is a strong surjective simulation
\[ h_L : L_1 \, || \cdots || \, L_n \to L \]
defined by
\[ h_L(\sigma) = \myvec{\sigma}, \qquad \sigma \in \bigcup_i \aL_i . \]
    
\item For $f : L \to L'$ in $\CPCVL$, the following square commutes:
\begin{diagram}
\GG\UU\Bi(L) & \rTo^{\GG\UU\Bi(f)} & \GG\UU\Bi(L') \\
\dTo^{h_L} & & \dTo_{h_{L'}} \\
L & \rTo_{f} & L'
\end{diagram}
\end{enumerate}
\end{theorem}

\begin{theorem}
    $\BB$ is left adjoint to $\GG$.
\end{theorem}

\begin{proof}
    The proof follows directly from the following diagrams:
 \[
 \begin{diagram}
 L & \rTo^{\mathbf{1}_L} & \GG\BB(L) \\
 & \rdTo_{f} & \ddash_{\GG(d, \phi)} \\
 & & \GG(A)
 \end{diagram}
 \qquad\qquad\qquad
 \begin{diagram}
 \BB(L) & \rTo^{F(f)} & \BB\GG(A) \\
 \ddash^{(d, \phi)} & \ldTo_{(\Id_{A}, g_{A})} & \\
 A & &
 \end{diagram}
 \]
    
\noindent The uniqueness of $\phi$ given $d$ follows from the fact that $F(L)$ is finitely generated; while the uniqueness of $d$ follows from the injectivity of $G$ on derivors (over a given Hom-set).
\end{proof}

\section{Final Remarks}

Firstly, we summarise some restrictions that have been imposed in the course of the development:

\begin{enumerate}[(i)]
    \item No empty languages or algebras.
    
    \item Constants must denote in $\PAlg(\Sigma)$/
    
    \item Derivors in $\CPAlg$ must:
        \begin{enumerate}[(I)]
            \item take $\epsilon$ to $\epsilon$
            \item not take a unary to a constant term
            \item be canonical derived homomorphisms.
        \end{enumerate}
\end{enumerate}

None of these restrictions seem particularly fundamental. From the order-theoretic point of view, the minimal element for $\PCL(\Sigma)$ we want is certainly $\{\epsilon\}$ rather than $\empset$ (assuming the subset ordering) --- otherwise recursive definitions involving concatenation would not get off the ground. Intuitively, we want the null \emph{behaviour}, i.e.~$\epsilon$,
rather than the null \emph{set} of behaviours. Disallowing empty algebras is also fairly standard. Of course, once we do so, (ii) is forced if $\Ac$ is to be total; while the first half of (i) forces (ii) if $G$ is to be total. As for (iii), parts (I) and (II) seem merely to reflect the greater generality of derived homomorphisms on unary algebras as compared to monoid homomorphisms; in particular, the possibility excluded by (II) would give language morphisms a sort of context-sensitive erasing capability. FInally, part (III) is needed to ensure that $G$ is injective on derivors over a given hom-set. This is used only for proving adjointness. What we are doing is to choose canonical names for extensionally equivalent objects (modulo $G$).

One particular issue of generalisation is exposed rather clearly by the algebraic framework. All the algebraic concepts used are valid for arbitrary signatures, and do not depend on the unary nature of the signature inherited from the language domain. The question arises, what computational significance can be read back into such a generalisation?

Finally, we make a rather general statement about our approach. Pnueli \cite{P77} has made a useful distinction between \emph{exogenous} and \emph{endogenous} logics of programs. The former, typified by Hoare-style proof systems for programming languages, are uniform proof theories valid for all programs in some language. The latter are logics tailored to the particular ``world'' of computations arising from a single program. We wish to make a similar distinction between 
exogenous and endogenous algebraic semantics. The former gives a uniform semantics for all the programs in a language, and develops an algebra \emph{of} the programs, in which the operators are the program-forming constructs. This is initial algebra semantics \cite{ADJ77}, exemplified for concurrent processes by Milner's CCS \cite{MM79,M80}. Endogenous algebraic semantics makes an algebra of each individual program or system. The operations that build more complex systems are operations on algebras, rather than on elements of an algebra. This is the style of algebraic semantics employed in this paper. Apart from its appearance in \cite{L70}, this style also has a great deal in common with the algebraic approach to abstract data types \cite{ADJ78, BG77}. Moreover, although exogenous semantics has the advantages of power and universality noted by Pneuli, we feel that endogenous semantics has significant advantages of its own, so that it is a style worth pursuing in tandem with the exogenous approach, rather than only by default. Specifically, these advantages include:

\begin{enumerate}[(i)]
    \item the opportunity to use ``problem-oriented'' signatures
    \item the possibility of applying the methodology of progressive transformations from high-level specification to implementation
    \item the ability to do most of the verification \emph{before} making all the decisions about representation.
\end{enumerate}

\end{document}